\documentclass[final,1p,times,authoryear]{elsarticle}
\usepackage{amsmath}
\usepackage{amsthm}
\usepackage{amssymb}
\usepackage{amsfonts}
\usepackage{tikz}
\usepackage[plainpages=false,pdfpagelabels,colorlinks=true,citecolor=blue,hypertexnames=false]{hyperref}
\newtheorem{theorem}{Theorem}

\newtheorem{prop}[theorem]{Proposition}
\newtheorem{corollary}[theorem]{Corollary}
\newtheorem{lemma}[theorem]{Lemma}

\newtheorem{defi}[theorem]{Definition}
\newtheorem{example}[theorem]{Example}

\def\ord{\operatorname{ord}}
\def\<#1>{\langle#1\rangle}
\def\defect{\operatorname{defect}}

\newcommand{\bQ}{ {\mathbb Q}}

\newcommand{\bC}{ {\mathbb C}}

\newcommand{\bN}{ {\mathbb N}}

\newcommand{\cX}{ {\mathcal X}}

\let\set\mathbb
\def\lc{\operatorname{lc}}

\def\lt{\operatorname{lt}}
\def\im{\operatorname{im}}

\def\diag{\operatorname{diag}}

\begin{document}

\begin{frontmatter}

  \title{Reduction-Based Creative~Telescoping for Fuchsian D-finite~Functions\footnote{S.\ Chen was supported by the NSFC grant
11501552 and by the President Fund of the Academy of
Mathematics and Systems Science, CAS (2014-cjrwlzx-chshsh).\\
M. van Hoeij was supported by the National Science Foundation (NSF) grant 1618657.\\
M. Kauers was supported by the Austrian Science Fund (FWF) grants F50-04 and Y464-N18.\\
C. Koutschan was supported by the Austrian Science Fund (FWF) grant W1214.
}}

  \author{Shaoshi Chen}
  \address{KLMM,\, AMSS, \,Chinese Academy of Sciences, Beijing, 100190, (China)}
  \ead{schen@amss.ac.cn}

  \author{Mark van Hoeij}
  \address{Department of Mathematics,  Florida State University, FL 32306-4510, (USA)}
  \ead{hoeij@math.fsu.edu}

  \author{Manuel Kauers}
  \address{Institute for Algebra, Johannes Kepler University, Altenberger Stra\ss e 69, A-4040 Linz, (Austria)}
  \ead{manuel.kauers@jku.at}

  \author{Christoph Koutschan}
  \address{RICAM, Austrian Academy of Sciences, Altenberger Stra\ss e 69, A-4040 Linz, (Austria)}
  \ead{christoph.koutschan@ricam.oeaw.ac.at}
  
\begin{abstract}
  Continuing a series of articles in the past few years on creative telescoping using reductions,
  we adapt Trager's Hermite reduction for algebraic functions to fuchsian D-finite functions and
  develop a reduction-based creative telescoping algorithm for this class of functions, thereby
  generalizing our recent reduction-based algorithm for algebraic functions, presented at ISSAC 2016.
\end{abstract}

\end{frontmatter}

\section{Introduction}\label{SECT:intro}

The classical question in symbolic integration is whether the integral of
a given function can be written in ``closed form''. In its most restricted form,
the question is whether for a given function~$f$ belonging to some domain $D$
there exists another function~$g$, also belonging to~$D$, such that $f=g'$. For
example, if $D$ is the field of rational functions, then for $f=1/x^2$ we can
find $g=-1/x$, while for $f=1/x$ no suitable $g$ exists. When no $g$ exists
in~$D$, there are several other questions we may ask. One possibility is to ask
whether there is some extension~$E$ of $D$ such that in $E$ there exists some
$g$ with $f=g'$. For example, in the case of elementary functions, Liouville's
principle restricts the possible extensions~$E$, and there are algorithms
which construct such extensions whenever possible. Another possibility is
to ask whether for some modification $\tilde f\in D$ of~$f$ there exists a $g\in
D$ such that $\tilde f=g'$. Creative telescoping leads to a question of this
type. Here we are dealing with domains~$D$ containing functions in several
variables, say $x$ and~$t$, and the question is whether there is a linear
differential operator~$P$, nonzero and free of~$x$, such that there exists a
$g\in D$ with $P\cdot f=g'$, where $g'$ denotes the derivative of $g$ with
respect to~$x$. Typically, $g$~itself has the form $Q\cdot f$ for some operator
$Q$ (which may be zero and need not be free of~$x$). In this case, we call $P$
a telescoper for~$f$, and $Q$ a certificate for~$P$.

Creative telescoping is the backbone of definite integration, because
$P\cdot f=(Q\cdot f)'$ implies, for instance, $P\cdot\int_0^1 f(x,t)dx = (Q\cdot f)(1) - (Q\cdot f)(0)$.
A telescoper $P$ for~$f$ thus gives rise to an annihilating operator for the
definite integral $F(t)=\int_0^1 f(x,t)dx$.
\begin{example}[\citep{manin1958}]
  The algebraic function
  \[
    f(x,t)=\frac1{\sqrt{x(x-1)(x-t)}}
  \]
  does not admit an elementary integral with respect to~$x$.
  However, we have $P\cdot f=(Q\cdot f)'$ for
  \[
    P=4(t-1)t\,D_t^2 + 4(2t-1)D_t+1,\quad Q=\frac{2x(x-1)}{t-x}.
  \]
  This implies
  \[
    P\cdot \int_0^1 f(x,t)dx = \Bigl[\frac{2x(x-1)}{t-x}f(x,t)\Bigr]_{x=0}^{x=1},
  \]
  so the integral $F(t)=\int_0^1\frac1{\sqrt{x(x-1)(x-t)}}dx$ satisfies the differential equation
  \[
    4(t-1)t\,F''(t) + 4(2t-1)F'(t) + F(t) = 0.
  \]
\end{example}
In the common case when the right-hand side
collapses to zero, we say that the integral has ``natural boundaries''. Readers not
familiar with creative telescoping are referred to the
literature~\citep{PWZbook1996,Zeilberger1990c,Zeilberger1991,Zeilberger1990,Koepf1998,kauers11}
for additional motivation, theory, algorithms, implementations, and applications. There are
several ways to find telescopers for a given $f\in D$. In recent years, an
approach has become popular which has the feature that it can find a telescoper
without also constructing the corresponding certificate. This is interesting
because certificates tend to be much larger than telescopers, and in some
applications, for instance when an integral has natural boundaries, only the telescoper is of interest. This approach was first
formulated for rational functions $f\in C(t,x)$ by~\cite{BCCL2010} and later
generalized to rational functions in several variables~\citep{bostan13,lairez16}, to
hyperexponential functions~\citep{bostan13a} and, for the shift case, to hypergeometric
terms~\citep{chen15a,huang16} and binomial sums~\citep{bostan16}.
At ISSAC'16, three of the present authors have given a version for algebraic functions~\citep{chen16}.
In the present article, we extend this algorithm to fuchsian D-finite functions.

The basic principle of the general approach is as follows. Assume that the
$x$-constants $\mathrm{Const}_x(D)=\{\,c\in D:c'=0\,\}$ form a field and that $D$
is a vector space over the field of $x$-constants. Assume further that there is
some $\mathrm{Const}_x(D)$-linear map $[\cdot]\colon D\to D$ such that for every
$f\in D$ there exists a $g\in D$ with $f-[f]=g'$. Such a map is called a
\emph{reduction.} For example, in $D=C(t,x)$ Hermite reduction~\citep{Hermite1872} decomposes
any $f\in D$ into $f=g'+h$ with $g,h\in D$ such that $h$ is a proper rational function
(i.e., the numerator degree is smaller than the denominator degree)
with a squarefree denominator. In this case, we can take $[f]=h$.
In order to find a telescoper, we can compute $[f]$, $[\partial_t\cdot f]$, $[\partial_t^2\cdot f]$, \dots,
until we find that they are linearly dependent over $\mathrm{Const}_x(D)$.
Once we find a relation
$p_0[f] + \cdots + p_r[\partial_t^r\cdot f] = 0$,
then, by linearity,
$[p_0 f + \cdots + p_r \partial_t^r\cdot f] = 0$,
and then, by definition of $[\cdot]$, there exists a $g\in D$ such that $(p_0+\cdots + p_r\partial_t^r)\cdot f=g'$.
In other words, $P=p_0+\cdots + p_r\partial_t^r$ is a telescoper.

There are two ways to guarantee that this method terminates. The first
requires that we already know for other reasons that a telescoper exists. The
idea is then to show that the reduction $[\cdot]$ has the property that when
$f\in D$ is such that there exists a $g\in D$ with $g'=f$, then $[f]=0$. If
this is the case and $P=p_0+\cdots+p_r\partial_t^r$ is a telescoper for~$f$,
then $P\cdot f$ is integrable in~$D$, so $[P\cdot f]=0$, and by linearity
$[f]$, \dots, $[\partial_t^r\cdot f]$ are linearly dependent over
$\mathrm{Const}_x(D)$. This means that the method won't miss any
telescoper. In particular, this argument has the nice feature that we are
guaranteed to find a telescoper of smallest possible order~$r$. This approach
was taken by~\cite{chen15a}. The second way consists in showing that the
$\mathrm{Const}_x(D)$-vector space generated by $\{\,[\partial_t^i\cdot
  f]:i\in\set N\,\}$ has finite dimension. This approach was taken
by~\cite{BCCL2010,bostan13a}. It has the nice additional feature that every
bound for the dimension of this vector space gives rise to a bound for the
order of the telescoper. In particular, it implies the existence of a
telescoper.

In this paper, we generalize Trager's Hermite reduction for algebraic
functions to fuchsian D-finite functions (Section~\ref{sec:hermite}), which
yields a reduction-based creative telescoping algorithm via the first approach
(Section~\ref{sec:canonic}). In Section~\ref{sec:polynomial}, we introduce a
second reduction, called polynomial reduction, which together with the Hermite
reduction gives rise to a reduction-based creative telescoping algorithm via
the second approach.  It also delivers a new bound for the order of the
telescoper of a fuchsian D-finite function, and in particular an independent
proof for its existence.

\section{Fuchsian D-finite Functions}

Throughout the paper, let $C$ be a field of characteristic zero. We consider linear differential operators
$L=\ell_0+\cdots+\ell_n\partial_x^n$ with $\ell_0,\dots,\ell_n$ belonging to some ring
$R$ containing~$C$. Typical choices for $R$ will be $C[x]$ or~$C(x)$.
When $\ell_n\neq0$, we say that $\ord(L):=n$ is the order of~$L$.

Let $R$ be a differential ring and write $'$ for its derivation.
We write $R[\partial_x]$ for the algebra consisting of all linear differential operators, together
with the usual addition and the unique non-commutative multiplication satisfying
$\partial_xa=a\partial_x+a'$ for all $a\in R$. We shall assume throughout that $C\subseteq\mathrm{Const}_x(R)$.
The algebra $R[\partial_x]$ acts on a differential $R$-module~$F$ via
\[
  (\ell_0+\ell_1\partial_x+\cdots+\ell_n\partial_x^n)\cdot f=
   \ell_0f + \ell_1f' + \cdots + \ell_n f^{(n)}.
\]
An element $y\in F$ is called a solution of an operator $L\in R[\partial_x]$ if
$L\cdot y=0$.

By $\bar C$ we denote some algebraically closed field containing~$C$ (not necessarily the smallest).
An operator $L$ of order~$n$ is called fuchsian at a point $a\in\bar C$ if
it admits $n$ linearly independent solutions in
\[
  \bar C[[[x-a]]] := \bigcup_{\nu\in C} (x-a)^\nu\bar C[[x-a]][\log(x-a)].
\]
It is called fuchsian at $\infty$ if it admits $n$ linearly independent solutions in
\[
  \bar C[[[x^{-1}]]] := \bigcup_{\nu\in C} x^{-\nu} \bar C[[x^{-1}]][\log(x)].
\]
It is simply called \emph{fuchsian} if it is fuchsian at all $a\in\bar C\cup\{\infty\}$.
Note that the exponents $\nu$ are restricted to~$C$, not to the larger field~$\bar C$.
For simplicity, the dependence on $C$ is not reflected in the notation.

Examples for fuchsian operators are operators that have a basis of algebraic
function solutions, the Gauss hypergeometric differential operator, or the
operator $L=x\partial_x^2-\partial_x$, whose solutions are $1$ and $\log(x)$.
However, the class of fuchsian D-finite functions
considered in this paper is not as rich as it may seem at first glance, because
we require that the operators under consideration should be fuchsian at all
points \emph{including infinity}. Functions
such as $\exp(x)$, $\sin(x)$, $\cos(x)$, Bessel functions, etc.\ are only
fuchsian at all finite points but not at infinity.

For a fixed fuchsian operator~$L$, we will consider the left $R[\partial_x]$-module
$A=R[\partial_x]/\<L>$, where $\<L>$ denotes the left ideal generated by~$L$ in
$R[\partial_x]$.  Then $1\in A$ is a solution of~$L$, because we have $L\cdot 1=L=0$
in~$A$. We can say that $A$ consists of all the ``functions'' $f$ which can be
obtained from a ``generic'' solution~$y$ of $L$ by applying some operator $P\in
R[\partial_x]$ to it. When $R$ is a field, then $A$ is an $R$-vector space of
dimension~$n=\ord(L)$, generated by $1,\partial_x,\dots,\partial_x^{n-1}$.

It is instructive to compare this setup to the situation for algebraic
functions. Comparing $A=R[\partial_x]/\<L>$ to an algebraic function field $R[Y]/\<M>$
(when $R$ is a field), our operator $L$ plays the role of the minimal
polynomial~$M$. In the algebraic case, $Y$~is a formal solution of the equation
$M=0$, similar as $1\in A$ is a formal solution of~$L$. Besides these formal
solutions there are, for each fixed $a\in\bar C$, exactly $\deg_Y(M)$ different
Puiseux series solutions of $M=0$ at places above~$a$. They correspond in the
differential setting to the series solutions of $L$ in $\bar C[[[x-a]]]$, which
generate a $\bar C$-vector space of dimension $\ord(L)$.

The exponents of an element $f\in A=R[\partial_x]/\<L>$ at a point $a\in\bar C\cup\{\infty\}$
are the values $\alpha$ such that one of the series in $\bar C[[[x-a]]]$ (or $\bar C[[[x^{-1}]]]$, respectively)
associated to $f$ has $(x-a)^\alpha\log(x-a)^\beta$ (or $(\frac1x)^\alpha\log(x)^\beta$) as initial term.
For an element $f\in A$, let $n_f$ be the minimal order of an operator
$\tilde{L}\in R[\partial_x]\setminus\{0\}$ with $\tilde{L}\cdot f=0$.
We say that $a$ is an ordinary point of $f$ if the set of exponents of $f$ at $a$
is $\{0,1,\dots,n_f-1\}$ and the solutions at $a$ do not involve logarithms.
There can be at most finitely many non-ordinary points; these are called the singular points.
The defect of $f\in A$ at $a\in\bar C\cup\{\infty\}$, denoted $\defect_a(f)$, is defined 
as the sum of the exponents of $f$ at $a$ minus $\sum_{k=0}^{n_f-1}k=\frac12n_f(n_f-1)$.
Then the Fuchs relation~\citep{schlesinger95,ince26} says that we have
\[
 \sum_{a\in\bar C\cup\{\infty\}} \defect_a(f) = n_f(1-n_f)
\]
for all $f\in A$. This relation is the counterpart of the B\'{e}zout relation in the algebraic case.
Note that when $a$ is an ordinary point, then $\defect_a(f)=0$, but $\defect_a(f)=0$ does in
general not imply that $a$ is an ordinary point.

In the context of creative telescoping, we let $\bar C$ be some algebraically
closed field containing the rational function field $K=C(t)$, and we
use $R=K(x)$ instead of~$C(x)$. Integration will always be with respect to~$x$, but
besides the derivation $\partial_x$ there is now also the derivation with respect
to~$t$. The notation $f'$ will always refer to the derivative of $f$
with respect to~$x$, not with respect to~$t$. In addition to the operator
algebra $R[\partial_x]$, we consider the
operator algebra $R[\partial_x,\partial_t]$, in which $\partial_x,\partial_t$
commute with each other (although they need not commute with elements of~$R$).

The action of $R[\partial_x]$ on $\bar C[[[x-a]]]$ or $\bar C[[[x^{-1}]]]$ is extended
to $R[\partial_x,\partial_t]$ by letting $\partial_t$ act coefficient-wise.
We further assume that the action of $R[\partial_x]$ on $A=R[\partial_x]/\<L>$ is
extended to an action of $R[\partial_x,\partial_t]$ on~$A$ in a way that is
compatible with the action of $R[\partial_x,\partial_t]$ on series domains.
This means that when $y\in\bar C[[[x-a]]]$ is a solution of~$L$ and $f$ is an
element of~$A$, so that $f\cdot y$ is an element of $\bar C[[[x-a]]]$, then we
want to have $(\partial_t\cdot f)\cdot y=\partial_t\cdot(f\cdot y)$, where
the $\cdot$ in $(\partial_t\cdot f)$ refers to the action of $R[\partial_x,\partial_t]$
on $A$ and the three other dots refer to the action on $\bar C[[[x-a]]]$.

If $u\in A$ is such that $\partial_t\cdot 1 = u$ and $U\in R[\partial_x]$ is such
that $u=U+\<L>$, then the annihilator $\mathfrak{a}\subseteq R[\partial_x,\partial_t]$
of $1\in A$ in $R[\partial_x,\partial_t]$ contains $L$ and $\partial_t-U$. We
therefore have $\dim\mathfrak{a}=0$, which is the usual
definition of D-finiteness in the case of several variables~\citep{Zeilberger1990,chyzak98,koutschan09,kauers14c}.
Pathological situations, where we also have $\dim\mathfrak{a}=0$ but 
$\mathfrak{a}$ does not have a basis of the form $\<L,\partial_t-U>$, are not
considered in this paper. 

\section{Integral Bases}

Trager's Hermite reduction for algebraic functions rests on the notion of
integral bases. The notion of integral bases has been generalized to D-finite
functions~\citep{kauers15b}, and an algorithm was also given there for
computing such bases. We recall here the relevant definitions and properties.

Although the elements of a generalized series ring $\bar C[[[x-a]]]$ are formal
objects, the series notation suggests certain analogies with complex
functions.  For simplicity, let us assume throughout that $C\subseteq\set
R$. Terms $(x-a)^\alpha\log(x-a)^\beta$ or $(\tfrac1x)^\alpha\log(x)^\beta$ are
called \emph{integral} if $\alpha\geq0$. A series in
$\bar C[[[x-a]]]$ or $\bar C[[[x^{-1}]]]$ is called integral if it only contains integral
terms. A non-integral series is said to have a \emph{pole} at the reference
point. Note that in this terminology also $1/\sqrt{x}$ has a pole
at~$0$, while $\log(x)$ does not; this convention differs slightly from the
default setting of~\citep[Ex.~2]{kauers15b}, but can be achieved by defining
the function~$\iota$~\citep[Def.~1]{kauers15b} accordingly. Note also that
the terminology only refers to $x$ but not to~$t$.

Integrality at $a\in\bar C$ is not preserved by differentiation,
but if $f$ is integral at~$a$, then so is $(x-a)f'$. On the other hand,
integrality at infinity is preserved by differentiation; we even have the
stronger property that when $f$ is integral at infinity, then not only $f'$
but also $xf'=(x^{-1})^{-1}f'$ is integral at infinity.

Let $K$ be some field with $C\subseteq K\subseteq\bar C$.
Let $L\in K(x)[\partial_x]$ be a fuchsian operator. An element $f\in A=K(x)[\partial_x]/\<L>$
is called (locally) integral at $a\in\bar C\cup\{\infty\}$ if for every solution $y$
of $L$ in $\bar C[[[x-a]]]$ or $\bar C[[[x^{-1}]]]$, respectively, the series $f\cdot y$ is
integral. $f$~is called (globally) integral if it is locally integral at every
$a\in\bar C$ (``at all finite places'').

For an element $f\in A$ to have a ``pole'' at $a\in\bar C\cup\{\infty\}$ means
that $f$ is not locally integral at~$a$; to have a ``double pole'' at $a$ means
that $(x-a)f$ (or $\frac1xf$ if $a=\infty$) is not integral; to have a ``double
root'' at $a$ means that $f/(x-a)^2$ (or $f/(\frac1x)^2=x^2f$ if $a=\infty$) is integral,
and so on.

The set of all globally integral elements $f\in A$ forms a $K[x]$-submodule of~$A$.
A basis $\{\omega_1,\dots,\omega_n\}$ of this module is called an \emph{integral basis}
for~$A$. \cite{kauers15b} proposed an algorithm which computes an integral basis
for a given~$A$. This algorithm is a generalization of van Hoeij's
algorithm~\citep{vanHoeij94} for computing integral bases of algebraic function
fields~\citep{trager84,Rybowicz:1991:ACI:120694.120715}.

For a fixed $a\in\bar C$, let $\bar C(x)_a$ be the ring of rational functions $p/q$
with $q(a)\neq0$, and write $\bar C(x)_\infty$ for the ring of all
rational functions $p/q$ with $\deg_x(p)\leq\deg_x(q)$.
Then the set of all $f\in A$ which are locally integral at some
fixed $a\in\bar C\cup\{\infty\}$ forms a $\bar C(x)_a$-module. A basis of this module is
called a \emph{local integral basis} at $a$ for~$A$. The algorithm given by~\cite{kauers15b}
for computing (global) integral bases computes local integral bases at finite
points as an intermediate step. By an analogous algorithm, it is also possible
to compute a local integral basis at infinity.

An integral basis $\{\omega_1,\dots,\omega_n\}$ is always also a $K(x)$-vector space
basis of~$A$. A key feature of integral bases is that they make poles explicit. Writing
an element $f\in A$ as a linear combination $f=\sum_{i=1}^n f_i\omega_i$ for some
$f_i\in K(x)$, we have that $f$ has a pole at $a\in\bar C$ if and only if at least one
of the $f_i$ has a pole there.

\begin{lemma}\label{lemma:1}
  Let $L$ be a fuchsian operator and let
  $\{\omega_1,\dots,\omega_n\}$ be a local integral basis of $A=K(x)[\partial_x]/\<L>$ at $a\in\bar C\cup\{\infty\}$.
  Let $f\in A$ and $f_1,\dots,f_n\in K(x)$ be such that $f=\sum_{i=1}^nf_i\omega_i$.
  Then $f$ is integral at $a$ if and only if each $f_i\omega_i$ is integral at~$a$.
\end{lemma}
\begin{proof}
  The direction ``$\Leftarrow$'' is obvious. To show ``$\Rightarrow$'', suppose
  that $f$ is integral at~$a$. Then there exist $\tilde{f}_1,\dots,\tilde{f}_n\in\bar C(x)_a$ such that
  $f=\sum_{i=1}^n\tilde{f}_i\omega_i$. Thus $\sum_{i=1}^n(\tilde{f}_i-f_i)\omega_i=0$, and then
  $\tilde{f}_i=f_i$ for all~$i$, because $\{\omega_1,\dots,\omega_n\}$ is a vector space basis of~$A$.
  As elements of $\bar C(x)_a$, the $f_i$ are integral at~$a$, and hence also all the $f_i\omega_i$
  are integral at~$a$.
\end{proof}

The lemma says in particular that poles of the $f_i$ in a linear combination
$\sum_{i=1}^n f_i\omega_i$ cannot cancel each other.

\begin{lemma}\label{lemma:e}
  Let $L$ be a fuchsian operator and let
  $\{\omega_1,\dots,\omega_n\}$ be an integral basis of $A=K(x)[\partial_x]/\<L>$.
  Let $e\in K[x]$ and
  $M=((m_{i,j}))_{i,j=1}^n\in K[x]^{n\times n}$ be such that
  \[
    e\,\omega_i'=\sum_{j=1}^n m_{i,j}\omega_j
  \]
  for $i=1,\dots,n$ and $\gcd(e,m_{1,1},\dots,m_{n,n})=1$.
  Then $e$ is squarefree.
\end{lemma}
\begin{proof}
  Let $a\in\bar C$ be a root of~$e$. We show that $a$ is not a multiple root.
  Since $\omega_i$ is integral, it is in particular locally integral at~$a$.
  Therefore $(x-a)\omega_i'$ is locally integral at~$a$.
  Since $\{\omega_1,\dots,\omega_n\}$ is an integral basis, it follows that
  $(x-a)m_{i,j}/e\in\bar C(x)_a$ for all~$i,j$.
  Because of $\gcd(e,m_{1,1},\dots,m_{n,n})=1$, no factor $x-a$ of $e$
  can be canceled by all the~$m_{i,j}$.
  Therefore the factor $x-a$ can appear in $e$ only once.
\end{proof}

\begin{lemma} \label{lemma:degM}
  Let $L$ be a fuchsian operator and let $\{\omega_1,\dots,\omega_n\}$ be a local integral
  basis at infinity of $A=K(x)[\partial_x]/\<L>$. Let $e\in K[x]$ and $M=((m_{i,j}))_{i,j=1}^n\in K[x]^{n\times n}$
  be defined as in Lemma~\ref{lemma:e}. Then $\deg_x(m_{i,j})<\deg_x(e)$ for all $i,j$.
\end{lemma}
\begin{proof}
  Since every $\omega_i$ is locally integral at infinity, so is every $x\,\omega_i'$.
  Since $\{\omega_1,\dots,\omega_n\}$ is an integral basis at infinity, it follows that
  $xm_{i,j}/e\in\bar C(x)_\infty$ for all~$i,j$. This means that $1+\deg_x(m_{i,j})\leq\deg_x(e)$
  for all~$i,j$, and therefore $\deg_x(m_{i,j})<\deg_x(e)$, as claimed.
\end{proof}

A $K(x)$-vector space basis $\{\omega_1,\dots,\omega_n\}$ of $A=K(x)[\partial_x]/\<L>$ is
called \emph{normal} at $a\in\bar C\cup\{\infty\}$ if there exist $r_1,\dots,r_n\in
K(x)$ such that $\{r_1\omega_1,\dots,r_n\omega_n\}$ is a local integral basis
at~$a$. \cite{trager84} shows for the case of algebraic function fields how to construct
an integral basis which is normal at infinity from a given integral basis and
a given local integral basis at infinity. The same procedure also applies
in the present situation. It works as follows.

Let $\{\omega_1,\dots,\omega_n\}$ be a global integral basis and $\{\nu_1,\dots,\nu_n\}$ be a local
integral basis at infinity. Let $m_{i,j}\in K(x)$ be such that
\[
 \omega_i = \sum_{j=1}^n m_{i,j}\nu_j.
\]
For each~$i$, let $\tau_i\in\set Z$ be the largest integer such that $x^{\tau_i}m_{i,j}$ has no pole at infinity for any~$j$.
Then each $x^{\tau_i}\omega_i$ is locally integral at infinity.
Let $B\in K^{n\times n}$ be the matrix obtained by evaluating $((x^{\tau_i}m_{i,j}))_{i,j=1}^n$ at infinity
(this is possible by the choice of~$\tau_i$).
If $B$ is invertible, then the $x^{\tau_i}\omega_i$ form a local integral basis at infinity and we are done.
Otherwise, there exists a nonzero vector $a=(a_1,\dots,a_n)\in K^n$ with $aB=0$.
Among the indices $\ell$ with $a_\ell\neq0$ choose one where $\tau_\ell$ is minimal,
and then replace $\omega_\ell$ by $\sum_{i=1}^n a_i x^{\tau_i-\tau_\ell}\omega_i$.
Note that the resulting basis is still global integral.
Repeating the process, it can be checked that the value of $\tau_1+\cdots+\tau_n$ strictly increases in each iteration.
According to the following lemma, the sum is bounded, so the procedure must terminate after a finite
number of iterations.

\begin{lemma}\label{lemma:bound-exps}
  Let $\{\omega_1,\dots,\omega_n\}$, $\{\nu_1,\dots,\nu_n\}$, and $\tau_1,\dots,\tau_n$ be as above.
  Let $N$ be the number of points $a\in\bar C$ where at least one of the $\omega_i$ does not have
  $n$ distinct exponents in~$\set N$ (i.e., $N$ counts the
  finite singular points of~$L$ that are not ``apparent'' singularities).
  Then
  \[
    \tau_1+\cdots+\tau_n \leq \tfrac12 n(n-1)(N-1).
  \]
\end{lemma}
\begin{proof}
  We show that when $\tau_i\in\set Z$ is such that $x^{\tau_i}\omega_i$ is locally integral at infinity,
  then $\tau_i\leq \frac12 (n-1)(N-1)$, for every~$i$.
  Let $n_i$ be the minimal order of an operator $L\in R[\partial_x]\setminus\{0\}$ with $L\cdot\omega_i=0$.
  By the Fuchs relation we have
  \[
    \sum_{a\in\bar C\cup\{\infty\}}\defect_a(\omega_i) = n_i(1-n_i),
  \]
  hence
  \[
    \defect_\infty(\omega_i) = n_i(1-n_i) - \sum_{a\in\bar C}\defect_a(\omega_i).
  \]
  When all exponents of the series associated to $\omega_i$ at $a\in\bar C$ form a subset of~$\set N$ of size~$n_i$, then
  $\defect_a(\omega_i)\geq0$. At all other points~$a$, of which there are at most~$N$ by assumption,
  we still have the estimate $\defect_a(\omega_i)\geq-\frac12n_i(n_i-1)$, because $\omega_i$ is integral at all finite places.
  It follows that
  \[
    \defect_\infty(\omega_i)\leq n_i(1-n_i) + \tfrac12 n_i(n_i-1) N = \tfrac12 n_i(n_i-1)(N-2).
  \]
  Next, for every $r\in\set Z$ we have $\defect_\infty(x^{-r}\omega_i) = rn_i+\defect_\infty(\omega_i)$.
  Moreover, if $\tau\in\set Z$ is such that $x^{\tau}\omega_i$ is integral at infinity, then we must have
  $\defect_\infty(x^{\tau}\omega_i)\geq -\frac12n_i(n_i-1)$, i.e.,
  \[
    \defect_\infty(\omega_i) - \tau n_i \geq -\tfrac12n_i(n_i-1),
  \]
  and hence,
  \begin{alignat*}1
    \tau &\leq \frac1{n_i}\Bigl( \frac12n_i(n_i-1) + \defect_\infty(\omega_i) \Bigr) \\
      &\leq \tfrac12 (n_i-1) + \tfrac12 (n_i-1)(N-2) = \tfrac12 (n_i-1)(N-1) \leq \tfrac12(n-1)(N-1)
  \end{alignat*}
  as claimed.
\end{proof}

Although normality is a somewhat weaker condition on a basis than integrality,
it also excludes the possibility that poles in the terms of a linear combination
of basis elements can cancel:

\begin{lemma}\label{lemma:3}
  Let $L$ be a fuchsian operator and let $\{\omega_1,\dots,\omega_n\}$ be a basis of $A=K(x)[\partial_x]/\<L>$
  which is normal at some $a\in\bar C\cup\{\infty\}$.
  Let $f=\sum_{i=1}^n f_i\omega_i$ for some $f_1,\dots,f_n\in K(x)$.
  Then $f$ has a pole at $a$ if and only if
  there is some $i$ such that $f_i\omega_i$ has a pole at~$a$.
\end{lemma}
\begin{proof}
  Let $r_1,\dots,r_n\in K(x)$ be such that $\{r_1\omega_1,\dots,r_n\omega_n\}$ is a
  local integral basis at~$a$. By $f=\sum_{i=1}^n
  (f_ir_i^{-1})(r_i\omega_i)$ and by Lemma~\ref{lemma:1}, $f$~is integral at~$a$ iff all
  $f_ir_i^{-1}r_i\omega_i=f_i\omega_i$ are integral at~$a$.
\end{proof}

We will mostly be using bases that are integral at every point in $\bar C\cup\{\infty\}$ except one.
For the case of algebraic functions, the reason is that the only algebraic functions which are
integral at all finite places and also at infinity are the constant functions (Chevalley's
theorem~\citep[p.~9, Cor.~3]{Chevalley1951}). The results of~\cite{chen16}
depend heavily on this fact. There is no analogous result for fuchsian D-finite functions: such
functions may be integral at all finite places and also at infinity without being constant.
It is easy to construct examples using the Papperitz symbol.

\begin{example}\label{example:integral-everywhere}
  The operator $L = 3 x(x^2-1) D_x^2 + 2 (3x^2 - 1) D_x\in\set Q(x)[\partial_x]$
  has three singular points $0,+1,-1$. Infinity is an ordinary point of~$L$.
  At all three singularities, there is one local solution starting with exponent~$0$
  and another starting with exponent~$1/3$, so all the solutions are integral
  everywhere according to our standard definition of integrality of generalized
  series.
\end{example}

Fortunately, we can still be sure that there are not too many such functions.

\begin{lemma}\label{lemma:nopoles-space}
  Let $A=K(x)[\partial_x]/\<L>$ for some fuchsian operator~$L$, let
  $\{\omega_1,\dots,\omega_n\}$ be a global integral basis of~$A$ which is
  normal at infinity, and let $\tau_1,\dots,\tau_n\in\set Z$ be such that
  $\{x^{\tau_1}\omega_1,\dots,x^{\tau_n}\omega_n\}$ is a local integral basis
  at infinity. Denote by $V$ the set of all $f\in A$ which are integral at all
  finite places and at infinity.  Then $V$ is a $K$-vector space of finite
  dimension, and $\{\,x^j\omega_i : i=1,\dots,n; j=0,\dots,\tau_i\,\}$ is a
  basis of~$V$.
\end{lemma}
\begin{proof}
  It is clear that $V$ is closed under taking $K$-linear combinations, so it is clearly a vector space.
  We show that $B=\{\,x^j\omega_i : i=1,\dots,n; j=0,\dots,\tau_i\,\}$ is a basis.

  Every $x^j\omega_i\in B$ is by definition integral at all finite places and because of $0\leq j\leq \tau_i$ also
  integral at infinity. Therefore $x^j\omega_i\in V$, and therefore $B$ generates a subspace of~$V$.

  Conversely, let $f\in V$ be arbitrary. Then $f$ is in particular integral at all finite places,
  and since $\{\omega_1,\dots,\omega_n\}$ is a global integral basis we can write $f=\sum_{i=1}^n p_i\omega_i$
  for some polynomials $p_1,\dots,p_n$.
  If we had $\deg(p_i)>\tau_i$ for some~$i$, then $p_i\omega_i$ would not be integral at infinity
  (by definition of the numbers~$\tau_i$), and then, because $\{\omega_1,\dots,\omega_n\}$ is normal
  at infinity, Lemma~\ref{lemma:3} implies that $f$ would not be integral at infinity.
  But $f$ is in $V$ and therefore integral at all points, including infinity.
  It follows that $\deg(p_i)\leq \tau_i$ for all~$i$, and therefore $f$ is a $K$-linear combination
  of elements of~$B$.
\end{proof}

\begin{corollary}
With the notation of Lemmas~\ref{lemma:bound-exps}
and~\ref{lemma:nopoles-space}, we have
\[
  \dim_{\set K}(V)\leq n\bigl(\tfrac12(n-1)(N-1)+1\bigr).
\]
\end{corollary}
\begin{proof}
Note that the exponents~$\tau_i$ in Lemma~\ref{lemma:nopoles-space} are the
same as in Lemma~\ref{lemma:bound-exps}; in the proof of the latter it was
shown that $\tau_i\leq\frac12(n-1)(N-1)$. The estimate on the dimension of~$V$
then follows immediately.
\end{proof}

\section{Hermite Reduction}\label{sec:hermite}

Hermite reduction was first introduced by Hermite~\citep{Hermite1872} for
rational functions.  This reduction was later extended to elementary
functions by Risch~\citep{Risch1969,Risch1970,ACA1992,bronstein98,BronsteinBook}, to algebraic functions by
Trager~\citep{trager84,ACA1992,bronstein98}, and to hyperexponential
functions~\citep{bostan13a}.  These generalizations are the key step in many
integration algorithms, including the telescoping algorithm presented in this
paper. It turns out that the Hermite reduction for fuchsian D-finite functions
is literally the same as Trager's reduction for algebraic functions.

We start with a technical lemma, which is needed later to ensure that the
Hermite reduction always works. The analogous statement for algebraic
functions and its proof can be found in~\citep[pp. 46--47]{trager84}; Trager's
proof for the algebraic case directly carries over and is reproduced here
only for the convenience of the reader.

Throughout this section, let $L\in K(x)[\partial_x]$ be a fuchsian operator
of order~$n$ and let $A=K(x)[\partial_x]/\<L>$.

\begin{lemma}\label{lemma:ibv}
Let $v\in K[x]$ be a squarefree polynomial and let $\{\omega_1,\ldots,\omega_n\}$
be a basis of~$A$ that is locally integral at all roots of~$v$.
For some integer $\mu>1$ we define $\psi_i:=v^\mu\left(v^{1-\mu}\omega_i\right)'$; then
$\{\psi_1,\ldots,\psi_n\}$ is a local integral basis at each root of~$v$.
\end{lemma}
\begin{proof}
By expanding $\psi_i=v\omega_i'-(\mu-1)v'\omega_i$ one sees that the $\psi_i$
themselves are integral at all roots of~$v$. Let now $a\in\bar{C}$ be an
arbitrary but fixed root of~$v$. We have to show that each $f\in A$ that is
integral at~$a$ can be expressed as a linear combination of the $\psi_i$ with
coefficients in $\bar C(x)_a$. To the contrary, assume that there exists an
integral element~$f$ that requires $x-a$ in the denominator of some
coefficient, i.e.,
\[
  f = \frac{1}{v} \sum_{i=1}^n c_i \psi_i \quad\text{with }c_i\in \bar{C}(x)_a
  \text{ and } c_i(a)\neq0 \text{ for some } i
\]
(here we use the fact that $v$ is squarefree).  Further let $g=\sum_{i=1}^n
c_i'\omega_i$, which is obviously integral. Then also their sum
\[
  f+g = v^{\mu-1} \sum_{i=1}^n \left(c_i\bigl(v^{1-\mu}\omega_i\bigr)'
  + c_i'v^{1-\mu}\omega_i \right) 
  = v^{\mu-1} \sum_{i=1}^n \bigl(c_iv^{1-\mu}\omega_i\bigr)'
\]
must be integral. Since $\{\omega_1,\ldots,\omega_n\}$ is an integral basis at~$a$,
there exists for each $i=1,\ldots,n$ a series solution $y_i\in\bar C[[[x-a]]]$
of~$L$ such that $\omega_i\cdot y_i$ involves a term
$T=(x-a)^\alpha\log(x-a)^\beta$ with $0\leq\alpha<1$ and $\beta\in\bN$.
Let now $i$ be an index such that $c_i(a)\neq0$;
this implies that $T$ appears in $(c_i\omega_i)\cdot y_i$.
Using the fact that the $\omega_i$ form a local integral basis, it follows by
Lemma~\ref{lemma:1} that $T$ is also present in $h\cdot y_i$ where
$h=\sum_{i=1}^n c_i\omega_i$. Let now $T$ be the dominant term of $h\cdot y_i$,
i.e., among all terms with minimal $\alpha$ the one with the largest exponent~$\beta$.
From
\begin{equation}\label{eq:T}
  \bigl((x-a)^{\mu-1} \partial_x (x-a)^{1-\mu}\bigr) \cdot T =
  (1-\mu+\alpha)(x-a)^{\alpha-1}\log(x-a)^\beta + \beta(x-a)^{\alpha-1}\log(x-a)^{\beta-1}
\end{equation}
it follows that $(x-a)^{\alpha-1}\log(x-a)^\beta$ is the dominant term of
$\bigl(v^{\mu-1}\partial_xv^{1-\mu}\bigr)\cdot (h\cdot y_i)$; here we use the
assumption that $\mu>1$, because for $\mu=1$ and $\alpha=0$ the
coefficient $(1-\mu+\alpha)$ in~\eqref{eq:T} is zero. This calculation reveals that
$v^{\mu-1}\bigl(v^{1-\mu}h\bigr)'=f+g$ is not integral at~$a$, which contradicts our assumption on
the integrality of~$f$. Hence $\{\psi_1,\ldots,\psi_n\}$ is a local integral
basis at~$a$.
\end{proof}

Let $\{\omega_1,\ldots,\omega_n\}$ be an integral basis for~$A$.
Further let $e,m_{i,j}\in K[x]$ ($1\leq i,j\leq n$) be such that
$e\omega_i'=\sum_{j=1}^n m_{i,j}\omega_i$ and
$\gcd(e,m_{1,1},m_{1,2},\ldots,m_{n,n})=1$ as in Lemma~\ref{lemma:e}. For describing
the Hermite reduction we fix an integrand $f\in A$ and represent it in the
integral basis, i.e., $f=\sum_{i=1}^n (f_i/D)\,\omega_i$ with
$D,f_1,\ldots,f_n\in K[x]$. The purpose is to find $g,h\in A$ such that
$f=g'+h$ and $h=\sum_{i=1}^n(h_i/D^\ast)\,\omega_i$ with $h_1,\ldots,h_n\in K[x]$
and $D^\ast$ denoting the squarefree part of~$D$.
As differentiating the $\omega_i$ can introduce
denominators, name\-ly the factors of~$e$, it is convenient to consider those
denominators from the very beginning on, which means that we shall assume
$e\mid D$. Note that $\gcd(D,f_1,\ldots,f_n)$ can then be nontrivial.

We now execute one step of the Hermite reduction, where the multiplicity
$\mu>1$ of some nontrivial squarefree factor~$v\in K[x]$ of $D$ is reduced.
Let $u\in K[x]$ be such that $D=uv^\mu$; it follows that $\gcd(u,v)=1$ and
$\gcd(v,v')=1$. We want to find $g_1,\ldots,g_n,h_1,\ldots,h_n\in K[x]$
such that
\begin{equation}\label{eq:hred}
  \sum_{i=1}^n \frac{f_i}{uv^\mu}\omega_i =
  \biggl(\sum_{i=1}^n\frac{g_i}{v^{\mu-1}}\omega_i\biggr)' +
  \sum_{i=1}^n \frac{h_i}{uv^{\mu-1}}\omega_i.
\end{equation}
By a repeated application of such reduction steps one can decompose any $f\in A$
as $f=g'+h$ where the denominators of the coefficients of $h$ are squarefree
and the coefficients of $g$ are proper rational functions.

In order to determine the unknown polynomials $g_1,\ldots,g_n$ in~\eqref{eq:hred},
clearing the denominator $uv^\mu$ yields
\begin{equation}\label{eq:clear}
  \sum_{i=1}^n f_i\omega_i = \sum_{i=1}^n \biggl( uvg_i'\omega_i +
  uv^\mu g_i\left(v^{1-\mu}\omega_i\right)' + vh_i\omega_i \biggr),
\end{equation}
and then this equation is reduced modulo~$v$:
\begin{equation}\label{eq:modv}
  \sum_{i=1}^n f_i\omega_i =
  \sum_{i=1}^n g_iuv^\mu\left(v^{1-\mu}\omega_i\right)' \mod v.
\end{equation}
By Lemma~\ref{lemma:ibv} and from $\gcd(u,v)=1$ it follows that
the elements $uv^\mu\left(v^{1-\mu}\omega_i\right)'$ form a local integral basis
at each root of~$v$, which implies that the coefficients $g_i$ are
uniquely determined modulo~$v$.

By Lemma~\ref{lemma:e} the polynomial~$e$ is squarefree and therefore $e\mid uv$;
hence we can write $uv=ew$ for some $w\in K[x]$. By rewriting the
derivatives of the $\omega_i$ in terms of the integral basis,
Equation~\eqref{eq:modv} turns into
\begin{align*}
  \sum_{i=1}^n f_i\omega_i
  &= \sum_{i=1}^n g_i \bigl( uv\omega_i' - (\mu-1)uv'\omega_i \bigr) \mod v\\
  &= \sum_{i=1}^n g_i\, \biggl( w\sum_{j=1}^n m_{i,j}\omega_j - (\mu-1)uv'\omega_i \biggr) \mod v.
\end{align*}
Comparing coefficients w.r.t.\ $\omega_1,\ldots,\omega_n$ yields a
system of linear equations over $K[x]/\<v>$ for the unknown
functions $g_1,\ldots,g_n$. This system has
a unique solution. For the coordinates of the solution vector $(g_1,\dots,g_n)$,
we can choose representatives in $K[x]$ whose degrees are less than~$\deg_x(v)$.

The remaining unknowns $h_1,\ldots,h_n$ are obtained by plugging
$g_1,\ldots,g_n$ into Equation~\eqref{eq:clear}.

\begin{example}\label{ex:hr}
We consider the fuchsian D-finite function
\[
  \frac{1}{x^2} \log\Bigl(\frac{1}{x^2}-1\Bigr)\sqrt{\frac{1+x}{1-x}}.
\]
This function is annihilated by the second-order differential operator
\[
  L = (x^2-1)^2 x^2 \partial_x^2 + (x^2-1) (x+1) (7x-5) x \partial_x + 
    8x^4+5x^3-11x^2-5x+4.
\]
Using the algorithm described of~\cite{kauers15b} we compute the following
integral basis $\{\omega_1,\omega_2\}$ for $A=\bC(x)[\partial_x]/\<L>$:
\[
  \bigl\{ (x-1) x^2,\, (x^2-1) (x-1) x^3 \partial_x + 2(x-1) x^4 \bigr\}
\]
For the differentiation matrix, a simple calculation yields
\begin{equation}\label{eq:dmat}
  \begin{pmatrix} \omega_1' \\[2pt] \omega_2' \end{pmatrix} =
  \frac{1}{e} \begin{pmatrix} (x-1) (x+2) & 1 \\[2pt] -x^3-x^2+5x-4 & x^2-x+2 \end{pmatrix}
  \begin{pmatrix} \omega_1 \\[2pt] \omega_2 \end{pmatrix}
\end{equation}
with $e=(x^2-1)x$. As the integrand corresponds to $1\in A$, its
representation in the integral basis is
\[
  f = \frac{1}{(x-1)x^2} \omega_1 = \frac{x+1}{ex} \omega_1;
\]
using the notation employed above, we have $D=(x^2-1)x^2$, $f_1=x+1$,
and $f_2=0$. Here we can only reduce the power of~$x$ in the denominator, so
we start with $u=x^2-1$, $v=x$, and $\mu=2$. Then Equation~\eqref{eq:modv} leads
to the following linear system for the unknowns $g_1$ and $g_2$:
\[
  \begin{pmatrix} x-1 & -x^3-x^2+5 x-4 \\ 1 & 3-x \end{pmatrix}
  \begin{pmatrix} g_1 \\ g_2 \end{pmatrix} =
  \begin{pmatrix} x+1 \\ 0 \end{pmatrix} \mod x.
\]
Thus we get $g_1=3$ and $g_2=-1$ and the final result of the Hermite reduction is
\[
f=
  \biggl(\, \underbrace{\vphantom{\frac{1}{(x)}}  \frac{3}{x} \omega_1 - \frac{1}{x} \omega_2}_{=g} \,\biggr)' +\>
  \underbrace{\frac{-x^2-x+3}{(x^2-1)x} \omega_1 - \frac{1}{(x^2-1)x} \omega_2}_{=h}.
\]
\end{example}

For the integration of algebraic functions, it is known that Hermite reduction itself often
takes less time than the construction of an integral basis. If Hermite reduction is applied to some
other basis, for instance the standard basis $\{1,y,\dots,y^{n-1}\}$, it
either succeeds or it runs into a division by zero.
\cite{bronstein98a} noticed that when a division by zero occurs,
then the basis can be replaced by some other basis that is a little closer to
an integral basis, just as much as is needed to avoid this particular division
by zero. After finitely many such basis changes, the Hermite reduction will
come to an end and produce a correct output. This variant is known as lazy
Hermite reduction. The same technique also applies in the present situation.

\section{The Canonical Form Property}\label{sec:canonic}

Recall from the introduction that reduction-based creative telescoping requires
some $K$-linear map $[\cdot]\colon A\to A$ with the property that
$f-[f]$ is integrable in $A$ for every $f\in A$. This is sufficient for the
correctness of the method, but additional properties are needed in order to
ensure that the method terminates.

As also explained already in the introduction, one possibility consists in
showing that $[f]=0$ whenever $f$ is integrable. For the special case of
algebraic functions, Trager showed that his Hermite reduction has this
property~\citep[p.~50, Thm.~1]{trager84}. Essentially the same argument
works for the fuchsian case, as we will show next. The main difference is
that in the case of algebraic functions, we can exploit that all the functions
which have no poles at either a finite place or at infinity are the constant
functions. As we have pointed out above, this is no longer true in the fuchsian
D-finite case, but in this case, we can still exploit that the space of all
the functions which have no poles at all (neither at finite points nor at infinity)
form a finite-dimensional vector space over~$K$.

\begin{lemma}\label{lemma:pole_at_inf}
Let $\{\omega_1,\dots,\omega_n\}$ be an integral basis for~$A$ that is normal at
infinity. Let $g=\sum_{i=1}^ng_i\omega_i\in A$ be such that all its
coefficients $g_i\in K(x)$ are proper rational functions. If an integral
element $f\in A$ has a pole at infinity, then also $f+g$ has a pole at
infinity.
\end{lemma}
\begin{proof}
Since $f$ is integral we can write it as
$f=f_1\omega_1+\cdots+f_n\omega_n$ with $f_i\in K[x]$.
If $f$ has a pole at infinity, there is at least one index~$i$
such that $f_i\omega_i$ has a pole at infinity.
Since $f_i$ is a polynomial and $g_i$ is a proper rational function,
$f_i$~and $f_i+g_i$ have the same exponent at infinity.
Therefore, if $f_i\omega_i$ has a pole at infinity, so does $(f_i+g_i)\omega_i$.
But then, by Lemma~\ref{lemma:3}, also $f+g=\sum_{i=1}^n(f_i+g_i)\omega_i$ has
a pole at infinity.
\end{proof}

\begin{theorem}\label{thm:intiff0}
  Suppose that $f\in A$ has at least a double root at infinity (i.e., every
  series in $\bar C[[[x^{-1}]]]$ associated to $f$ only contains monomials
  $(1/x)^\alpha\log(x)^\beta$ with $\alpha\geq2$).
  Let $W=\{\omega_1,\dots,\omega_n\}$ be an integral basis for $A$ that is normal at infinity,
  and let $f=g'+h$ be the  result of the Hermite reduction with respect to~$W$.
  Let $V\subseteq A$ be the $K$-vector space of all elements that are integral at all places,
  including infinity, and let $U=\{v':v\in V\}$ be the space of all elements of~$A$
  that are integrable in~$V$.
  Then $f$ is integrable in~$A$ if and only if $h\in U$.
\end{theorem}

\begin{proof}
The direction ``$\Leftarrow$'' is trivial. To show the implication
``$\Rightarrow$'' assume that $f$ is integrable in~$A$. From $f=g'+h$ it follows that
then also $h$ is integrable in~$A$; let $H\in A$ be such that $H'=h$.  In order to show
that $h\in U$, we show that $H\in V$, i.e., we show that $H$ has no finite poles and
no poles at infinity.

It is clear that $H$ has no finite poles because $h$ has at most simple poles
(i.e., all series associated to $h$ have only exponents $\alpha\geq-1$).
This follows from the facts that the $\omega_i$ are integral and that
the coefficients of~$h$ have squarefree denominators.

If $H$ has a pole at infinity, then by Lemma~\ref{lemma:pole_at_inf} also
$g+H$ must have a pole at infinity, because Hermite reduction produces
$g=\sum_i g_i\omega_i$ with proper rational functions~$g_i$.  On the other
hand, since $f=g'+h=(g+H)'$ has at least a double root at infinity by
assumption, $g+H$ must have at least a single root at infinity. This is
a contradiction.
\end{proof}

By Lemma~\ref{lemma:nopoles-space}, the vector space~$V$ has finite dimension and computing a
vector space basis of it is not harder than computing an integral basis. Once
a basis $\{b_1,\dots,b_d\}$ of $V$ is known, it is also easy to obtain a basis
of~$U$, as this space is generated by $\{b_1',\dots,b_d'\}$. Therefore, we
can decide whether a given $f\in A$ with a double pole at infinity is integrable in~$A$
by first executing Hermite reduction, and then checking whether the Hermite
remainder~$h$ belongs to~$U$. More generally, by performing a reduction modulo~$U$
as a post-processing step after the Hermite reduction, we can even ensure that $f$
is integrable if and only if $h=0$.

\begin{example}
Let $L = 3 (x^3-x) D_x^2 + 2 (3x^2 - 1) D_x$ be the same operator as in
Example~\ref{example:integral-everywhere}; its solution space is spanned by
$y_1(x)=1$ and $y_2(x)=x^{1/3}{}\,_2F_{\!1}\bigl(\frac16,\frac23;\frac76;x^2\bigr)$.
An integral basis for $A=\bQ(x)[\partial_x]/\<L>$, that is also normal at
infinity, is given by $\omega_1=1$ and $\omega_2=(x^3-x)\partial_x$. Recall
that both solutions are integral everywhere, and hence $\omega_1\in V$.
Actually, the $K$-vector space~$V$ is spanned by~$\omega_1$, as can be seen
from the fact that $\tau_1=0$ and $\tau_2=-1$. The straightforward calculation
\[
  W' = \frac{1}{e}MW =
  \frac{1}{x^3-x}\begin{pmatrix}0 & 1\\ 0 & x^2-\frac13\end{pmatrix}W
  \quad\text{for }W=\begin{pmatrix}\omega_1\\ \omega_2\end{pmatrix}.
\]
exhibits that $e=x^3-x$. Consider now
\[
  f = \frac{3}{x^2}\omega_1 + \frac{2(2x+1)}{(x^3-x)^2}\omega_2,
\]
which has a double root at infinity. The result of the Hermite reduction is
\[
  f = \left(-\frac{3}{x}\omega_1-\frac{3(2x+1)}{2(x^3-x)}\omega_2\right)'
    - \frac{3}{x^3-x}\omega_2,
\]
which has a nonzero remainder. According to Theorem~\ref{thm:intiff0}, $f$
is integrable if this remainder lies in the subspace~$U=\{v':v\in V\}$.
Using the matrix~$M$ above, we find that $\omega_1'=\frac{1}{x^3-x}\omega_2$,
which is indeed a scalar multiple of our remainder. Hence, $f$ is integrable:
\[
  f = \left(-\frac{3(x+1)}{x}\omega_1 - \frac{3(2x+1)}{2(x^3-x)}\omega_2\right)'.
\]
\end{example}

Note that the condition in Theorem~\ref{thm:intiff0} that $f$ has a double
root at infinity is not a restriction, as it can always be achieved by
a suitable change of variables. Let $a\in C$ be an ordinary point of~$L$; by
the substitution $x\to a+1/x$ the ordinary point~$a$ is moved to
infinity. From
\[
  \int f(x) \,\mathrm{d}x = \int f\left(\frac{1}{x}+a\right)\left(-\frac{1}{x^2}\right) \mathrm{d}x
\]
we see that the new integrand has a double root at infinity.

Moreover, since the action of $\partial_t$ on series domains is defined coefficient-wise,
it follows that when $f$ has at least a double root at infinity (with respect to~$x$),
this is also true for $\partial_t\cdot f, \partial_t^2\cdot f, \partial_t^3\cdot f,\dots$,
and then also for every $K$-linear combination $p_0f+p_1\partial_t\cdot f+\cdots+p_r\partial_t^r\cdot f$.
Thus Theorem~\ref{thm:intiff0} implies that $p_0+p_1\partial_t+\cdots+p_r\partial_t^r$
is a telescoper for $f$ if \emph{and only if}
$[p_0f+p_1\partial_t\cdot f+\cdots+p_r\partial_t^r\cdot f]=0$.

We already know for other reasons~\citep{Zeilberger1990,chyzak00} that
telescopers for D-finite functions exist, and therefore the reduction-based
creative telescoping procedure with Hermite reduction with respect to an
integral basis that is normal at infinity plus reduction modulo $U$
as reduction function $[\cdot]$ succeeds when applied to an integrand $f\in A$ that has a double root at infinity.
In particular, the method finds a telescoper of smallest possible order.
Again, if $f$ has no double root at infinity, we can produce one by a change of variables.
Note that a change of variables $x\to a+1/x$ with $a\in C$ has no effect on
the telescoper.

\section{Polynomial Reduction}\label{sec:polynomial}

Recall that instead of requesting that $[f]=0$ if and only if $f$ is integrable
(first approach), we can also justify the termination of reduction-based
creative telescoping by showing that the $K$-vector space generated by $\bigl\{\,[\partial_t^if]:i\in\set N\,\bigr\}$
has finite dimension (second approach). If $[\cdot]$ is just the Hermite
reduction, we do not necessarily have this property. We therefore introduce below an
additional reduction, called \emph{polynomial reduction,} which we apply after
Hermite reduction. We then show that the combined reduction (Hermite reduction
followed by polynomial reduction) has the desired dimension property for the
space of remainders. As a result, we obtain a new bound on the order of the
telescoper, which is similar to that of~\cite{chen14a}.

In this approach, we use two integral bases. First we use a global integral basis
in order to perform Hermite reduction. Then we write the
remainder $h$ with respect to some local integral basis at infinity and perform the
polynomial reduction on this representation.

Throughout this section let $W=(\omega_1,\ldots,\omega_n)^T\in A^n$ be such
that $\{\omega_1, \ldots, \omega_n\}$ is a global integral basis of~$A$, and
let $e\in K[x]$ and $M=((m_{i,j}))\in K[x]^{n\times n}$ be such that $eW'=MW$
and $\gcd(e, m_{1, 1}, m_{1, 2}, \ldots, m_{n ,n})=1$. The Hermite reduction
described in Section~\ref{sec:hermite} decomposes an element $f\in A$ into
the form
\[
  f = g' + h = g' + \sum_{i=1}^n \frac{h_i}{de} \omega_i,\qquad
  g, h\in A,
\]
with $h_i, d\in K[x]$ such that $\gcd(d, e)=\gcd(d,h_1,\dots,h_n)=1$ and $d$ is squarefree.
\begin{lemma}\label{LEM:d}
  Let $h$ be as above. If $h$ is integrable in~$A$,
  then its integral is integral, and we have $d\in K$.
\end{lemma}
\begin{proof}
  Suppose that $h$ is integrable in~$A$, i.e., there exist $b_i\in K(x)$
  such that $h = \bigl(\sum_{i=1}^n b_i \omega_i\bigr)'$.
  If one of the $b_i$ had a pole at a finite place, then $\sum_{i=1}^n b_i\omega_i$
  would have a pole at a finite place, because $\{\omega_1,\dots,\omega_n\}$ is
  an integral basis. But then $h$ would have a pole of order greater than~$1$
  there, which is impossible because $\gcd(d,e)=1$ and $d$ is squarefree and
  $\{\omega_1,\dots,\omega_n\}$ is an integral basis.
  Therefore, $b_1,\dots,b_n\in K[x]$ and we have shown that the integral of $h$
  is integral. The claim on $d$ then follows directly from the definition of~$e$.
\end{proof}

Note that the lemma continues to hold when $\{\omega_1,\dots,\omega_n\}$ is
a local integral basis at a finite place~$a\in\bar C$. In this case, we can conclude
that the integral is locally integral at $a$ and $d\in K(x)_a$. 

By the extended Euclidean algorithm, we compute $r_i, s_i\in K[x]$ such that
$h_i = r_i e + s_i d$ and $\deg_x(r_i) < \deg_x(d)$. Then the Hermite remainder~$h$
decomposes as
\begin{equation}\label{EQ:h}
  \sum_{i=1}^n \frac{h_i}{de}\omega_i =
  \sum_{i=1}^n \frac{r_i}{d}\omega_i +
  \sum_{i=1}^n \frac{s_i}{e}\omega_i.
\end{equation}
We now introduce the \emph{polynomial reduction} whose goal is to confine the $s_i$ to a finite-dimensional
vector space over~$K$. Similar reductions have been introduced and used in creative telescoping
for hyperexponential functions~\citep{bostan13a} and hypergeometric terms~\citep{chen15a}.
Our version below is slightly different from these, and also from the reduction given by~\cite{chen16} for
the algebraic case, because we will be considering Laurent polynomials instead of polynomials.
Note that the same idea can be applied in the polynomial reduction for $q$-hypergeometric terms~\cite{du16}.

Throughout the rest of the section, let $V = (\nu_1, \ldots, \nu_n)^T\in A^n$ be such that its entries
form a local integral basis at $\infty$ as well as a local integral basis at every $a\in\bar C\setminus\{0\}$,
and which is normal at~$0$.
The existence of such a basis follows from the existence of global integral bases that are normal at infinity,
as follows.

\begin{lemma}\label{LM:CB}
Let $W =\{\omega_1, \ldots, \omega_n\}$ be an integral basis of~$A$ that is normal at infinity. Then
there exist integers $\tau_1, \ldots, \tau_n\in\set Z$ such that
$V := \{\nu_1, \ldots, \nu_n\}$ with $\nu_i = x^{\tau_i} \omega_i$ ($i=1,\dots,n$)
is a basis of~$A$ which is normal at $0$ and integral at all other points (including infinity).
\end{lemma}
\begin{proof}
It is clear that such a basis $V$ will be normal at zero, because multiplying the generators by
the rational functions $x^{-\tau_i}$ brings it back to a global integral basis, which is in particular
a local integral basis at zero.
It is also clear that such a basis will be integral at every other point $a\in\bar C\setminus\{0\}$, because the
multipliers $x^{\tau_i}$ are locally units at such~$a$.
Finally, since the original basis is normal at infinity, there exist rational functions $u_1,\dots,u_n$
such that $\{u_1\omega_1,\dots,u_n\omega_n\}$ is a local integral basis at infinity.
Since $u_i$ can be written as $u_i=x^{\tau_i}\tilde{u}_i$ with $\tau_i\in\set Z$ and $\tilde{u}_i$ being a unit
in $\bar{C}(x)_\infty$, we see that also $V$ is a local integral basis at infinity.
\end{proof}

In the case of algebraic functions~\citep{chen16}, it is a consequence of Chevalley's theorem that the exponents
$\tau_i$ can never be positive, and that they can be zero only if the function is constant.
In the more general fuchsian situation, this is no longer the case. 
Let $a\in K[x]$ and $B = ((b_{i, j}))\in K[x]^{n \times n}$ be such that $aV'=BV$ and
$\gcd(a, b_{1, 1}, b_{1, 2}, \ldots, b_{n ,n})=1$. By Lemma~\ref{LEM:d}, we may assume that
$a=x^\lambda e$ for some $\lambda\in\set N$, and we will do so. Writing the expression
$\sum_{i=1}^n\frac{s_i}{e}\omega_i$ from equation~\eqref{EQ:h} in terms of the basis~$V$,
we obtain $\sum_{i=1}^n\frac{x^{\lambda-\tau_i}s_i}{x^\lambda e}\nu_i$, where now
the numerators of the coefficients are Laurent polynomials. Note however that since
$s_1,\dots,s_n$ are in~$K[x]$, the new numerators cannot have arbitrarily negative exponents.
In fact, they will live in $x^{-\tau}K[x]$ where $\tau=\max(\tau_1,\dots,\tau_n)$.
An a priori bound for $\tau$ follows from Lemma~\ref{lemma:bound-exps}.

The purpose of polynomial reduction is to write
\[
  \sum_{i=1}^n\frac{x^{\lambda-\tau_i}s_i}{x^{\lambda}e}\nu_i=\tilde g'+
  \sum_{i=1}^n\frac{\tilde s_i}{x^\lambda e}\nu_i,
\]
where the $\tilde s_i$ belong to a finite-dimensional subspace of~$x^{-\tau}K[x]$.
Let us write $K[x]_{\eta,\mu}$ for the subspace of $K[x,x^{-1}]$ consisting of all Laurent polynomials
whose exponents are at least $\eta$ and at most~$\mu$. Then $x^{-\tau}K[x]=K[x]_{-\tau,\infty}$, and
while $x^{-\tau_1}s_1,\dots,x^{-\tau_n}s_n$ belong to this space, we will show that it is possible to
choose $\tilde s_1,\dots,\tilde s_n$ that belong to $K[x]_{-\tau,\delta}$ for some finite~$\delta\in\set Z$.
To this end, note that for any $P = (p_1, \ldots, p_n)\in K[x,x^{-1}]^n$ we have
\begin{equation} \label{EQ:polyred}
  (PV)' = \sum_{i=1}^n (p_i \nu_i)' = \frac{x^\lambda eP' + PB}{x^\lambda e}\,V.
\end{equation}
This motivates the following definition.

\begin{defi}
  Let 
  \[
    \phi_V\colon K[x]_{-\tau+1,\infty}^n \to K[x]_{-\tau,\infty}^n,\quad
    \phi_V(P) = x^\lambda eP' + PB.
  \]
  We call $\phi_V$ the \emph{map for polynomial reduction} with respect to~$V$, and call
  the subspace
\[
 \im(\phi_V) = \bigl\{\phi_V(P) \mid P \in K[x]_{-\tau+1,\infty}^n\bigr\}\subseteq K[x]_{-\tau,\infty}^n
\]
the \emph{subspace for polynomial reduction} with respect to~$V$.
\end{defi}

Note that, by construction and because of Lemma~\ref{LEM:d}, $Q\in K[x]_{-\tau,\infty}^n$ belongs to
$\im(\phi_V)$ if and only if $\frac{1}{x^\lambda e}QV$ is integrable in~$A$. 

We can always view an element of $K[x]_{\eta,\mu}^n$ (resp. $K[x]_{\eta,\mu}^{n\times n}$) as a Laurent polynomial in~$x$
with coefficients in~$K^n$ (resp. $K^{n\times n}$). In this sense we use the notation $\lc(\cdot)$
for the leading coefficient and $\lt(\cdot)$ for the leading term of a vector (resp. matrix).
For example, if $P\in K[x]_{\eta,\mu}^n$ is of the form
\[
  P = p^{(i_1)}x^{i_1} + \dots + p^{(i_m)}x^{i_m},\quad p^{(i)}\in K^n,\quad i_1\leq\cdots\leq i_m,\quad p^{(i_m)}\neq0
\]
then $\deg_x(P)=i_m$, $\lc(P)=p^{(i_m)}$, and $\lt(P)=p^{(i_m)}x^{i_m}$.

Let $\{e_1, \ldots, e_n\}$ be the standard basis of~$K^n$.
Then the $K$-vector space $K[x]_{\eta,\mu}^n$ is generated by
\[
  \cX_{\eta,\mu}:= \bigl\{e_ix^j \mathrel{|} 1\leq i \leq n,\, \eta\leq j\leq \mu\bigr\}.
\]

\begin{defi}
  Let $N_V$ be the $K$-subspace of $K[x]_{-\tau,\infty}^n$ generated by
\[
  \bigl\{t \in \cX_{-\tau,\infty} \mathrel{|} t \neq \lt(P) \ \text{for all $P\in \im(\phi_V)$}\bigr\}.
\]
Then $K[x]_{-\tau,\infty}^n = \im(\phi_V) \oplus N_V$.
We call $N_V$ the \emph{standard complement} of $\im(\phi_V)$.
For any $P\in K[x]_{-\tau,\infty}^n$, there exist $P_1\in K[x]_{-\tau+1,\infty}^n$ and~$P_2\in N_V$
such that $P=\phi_V(P_1)+P_2$ and
\[
  \frac{1}{x^\lambda e}PV = (P_1V)' + \frac{1}{x^\lambda e}P_2V.
\]
This decomposition is called the \emph{polynomial reduction} of~$P$
with respect to~$V$.
\end{defi}

\begin{prop}\label{PROP:finite}
Let $\lambda\in\set N$, $e\in K[x]$ and $B\in K[x]^{n \times n}$ be such that $x^\lambda eV'=BV$, as before.
If $\deg_x(B) \leq \lambda+\deg_x(e)-1$, then $N_V$ is a finite-dimensional $K$-vector space.
\end{prop}
\begin{proof}
For brevity, let $\delta:=\lambda+\deg_x(e)-1$. We distinguish two cases.

\smallskip
{\it Case 1.}
Assume that $\deg_x(B) < \delta$. For any $P\in K[x]_{-\tau+1,\infty}^n$ of degree~$\mu$, we have
\[
  \lt\bigl(\phi_V(P)\bigr) = \mu\lc(e)\lc(P)x^{\mu+\delta}.
\]
Thus all monomials $e_i x^j$ with $1\leq i\leq n$ and $j\geq \delta+1$ are not in~$N_V$,
and thus $\dim N_V\leq n(\tau+\delta+1)<\infty$.

\smallskip
{\it Case 2.}
Assume that $\deg_x(B)=\delta$. For any $P\in K[x]_{-\tau+1,\infty}^n$ of degree~$\mu$, we have
\[
  \lt\bigl(\phi_V(P)\bigr) = \lc(P)(\mu\lc(e)I_n + \lc(B))x^{\mu+\delta}.
\]
Let $\ell$ be the largest nonnegative integer such that $-\ell \lc(e)$ is an
eigenvalue of $\lc(B)\in K^{n\times n}$ (or let $\ell=0$ if no such integers exist).
Then for any $\mu>\ell$, the matrix $\mu\lc(e)I_n + \lc(B)\in K^{n\times n}$ is invertible.
So any monomial $e_ix^j$ with $j> \ell+\delta$ is not in~$N_V$ for any $i=1, \ldots, n$,
and thus $\dim N_V\leq n(\tau+\delta+\ell+1)<\infty$.
\end{proof}

It follows from our general assumptions on $V$ that the
condition $\deg_x(B) \leq \lambda+\deg_x(e)-1$ is always satisfied.
Therefore, by the combination of Hermite reduction described in Section~\ref{sec:hermite}
with polynomial reduction, we get the following theorem.

\begin{theorem}\label{THM:polyred}
Let $W\in A^n$ be an integral basis of~$A$ that is normal at infinity.
Let
\[
  T = \diag\bigl(x^{\tau_1}, \ldots, x^{\tau_n}\bigr) \in K(x)^{n\times n}
\]
be such that $V := TW$ is a local integral basis at infinity.  Let $e\in K[x]$,
$\lambda \in \bN$, and $B, M \in K[x]^{n \times n} $ be such that $eW' = MW$
and $x^\lambda eV' = BV$.  Then any element $f\in A$ can be decomposed into
\begin{equation}\label{EQ:add}
  f = g' + \frac{1}{d} RW + \frac{1}{x^\lambda e} QV,
\end{equation}
where $g\in A$, $d\in K[x]$ is squarefree and $\gcd(d, e)=1$, $R, Q\in K[x]^n$
with $\deg_x(R) < \deg_x(d)$ and $Q\in N_V$, which is a finite-dimensional
$K$-vector space. Moreover, $R=Q=0$ if and only if $f$ is integrable in~$A$.
\end{theorem}
\begin{proof}
After performing the Hermite reduction on $f$, we get
\begin{equation}\label{EQ:add1}
  f = \tilde{g}' + \frac{1}{d} RW + \frac{1}{e} SW,
\end{equation}
where~$R = (r_1, \ldots, r_n)\in K[x]^n$ and $S = (s_1, \ldots, s_n)\in K[x]^n$
with~$r_i, s_i$ introduced in~\eqref{EQ:h}. By Lemma~\ref{LM:CB}, there exists
$T = \diag(x^{\tau_1}, \ldots, x^{\tau_n}) \in K(x)^{n\times n}$ such that $V = TW$
is a local integral basis at infinity. By the same lemma it follows
that $V$ is also normal at $0$ and integral at all other points.
The derivative of~$V$ can be reexpressed in the same basis as
\[
  V' = (TW)' = \biggl(T' + \frac{1}{e}TM\biggr)T^{-1}V = \frac{1}{x^\lambda e}BV,
\]
Since $V$ is a local integral basis at infinity, it follows by
Lemma~\ref{lemma:degM} that $\deg_x(B) \leq \lambda + \deg_x(e)-1$, which is a
prerequisite for applying Proposition~\ref{PROP:finite}; hence the
corresponding space~$N_V$ is finite-dimensional, as claimed.
We rewrite the last summand in~\eqref{EQ:add1} w.r.t. the new basis~$V$:
\[
  \frac{1}{e} SW = \frac{1}{x^\lambda e} \tilde{S}V,
\]
where $\tilde{S} = x^\lambda S T^{-1} \in K[x,x^{-1}]^n$. Note that the entries
of~$\tilde{S}$ are not necessarily polynomials, but Laurent polynomials in~$x$.
Indeed, because of Lemma~\ref{lemma:nopoles-space}, some of the $\tau_i$
may be positive, and actually, $\tilde{S}\in K[x]_{-\tau,\infty}^n$ where
$\tau:=\max\{\tau_1,\dots,\tau_n\}$, as before. Next, using the polynomial reduction, we decompose
$\tilde{S}$ into $\tilde{S} = \phi_{V}(S_1) + S_2$ with
$S_1\in K[x]_{-\tau+1,\infty}^n$ and $S_2\in N_V$, which means that
\[
  \frac{1}{e} SW = (S_1 V)' + \frac{1}{x^\lambda e} S_2 V.
\]
We finally obtain the decomposition~\eqref{EQ:add} by setting
$g = \tilde{g} + S_1 V$ and $Q = S_2$.

For the last assertion, assume that $f$ is integrable (the other direction of
the equivalence holds trivially). Then Lemma~\ref{LEM:d} implies that $d\in K$,
and therefore $R$ must be zero because $\deg_x(R) < \deg_x(d)$. Hence the last
summand in~\eqref{EQ:add} is also integrable, i.e., there exist $c_i\in K(x)$
such that
\[
  \frac{1}{x^\lambda e} QV = \Biggl(\sum_{i=1}^n c_i \nu_i\Biggr)'.
\]
Note that the expression on the left-hand side has only simple poles at finite
points except~$0$. Therefore, by Lemma~\ref{LEM:d}, its integral is integral
at all nonzero finite points. In other words, the coefficients $c_i$ are actually
Laurent polynomials in $K[x,x^{-1}]$, which implies that $Q \in \im(\phi_V)$.
Since $\im(\phi_V) \cap N_V = \{0\}$, it follows that $Q=0$.
\end{proof}
The decomposition in~\eqref{EQ:add} is called an \emph{additive decomposition} of~$f$ with respect to~$x$.
We now discuss how to compute telescopers for elements of~$A$ via Hermite reduction and
polynomial reduction.

We first consider the additive decompositions of the successive derivatives $\partial_t^i\cdot f$ for $i\in \bN$.
Assume that
\begin{equation}\label{EQ:tmat}
  \partial_t\cdot W = \frac{1}{\tilde{e}} \tilde{M}W
  \quad \text{and}\quad
  \partial_t\cdot V = \frac{1}{\tilde{a}} \tilde{B}V,
\end{equation}
for some polynomials $\tilde{e},\tilde{a}\in K[x]$ and matrices
$\tilde{M},\tilde{B}\in K[x]^{n\times n}$ such that $\tilde{e}$ is coprime
with~$\tilde{M}$ and $\tilde{a}$ is coprime with~$\tilde{B}$.  By
\cite[Prop.~7]{chen14a}, we have that $\tilde e \mid e$ and $\tilde{a} \mid
x^\lambda e$.  Hence, we can take $\tilde{e} = e$ and $\tilde{a} = x^\lambda
e$ in~\eqref{EQ:tmat}, by multiplying the matrices $\tilde{M}$ and $\tilde{B}$
by some factors of~$x^\lambda e$. Now we differentiate~\eqref{EQ:add} with
respect to~$t$, and obtain, after a direct calculation, $\partial_t\cdot f =
(\partial_t\cdot g)' + h$, where
\[
  h =
    \left(\partial_t\cdot\biggl(\frac{1}{d}R\biggr) + \frac{1}{de}R\tilde{M}\right)W +
    \left(\partial_t\cdot\biggl(\frac{1}{x^\lambda e}Q\biggr) + \frac{1}{x^{2\lambda} e^2}Q\tilde{B}\right)V.
\]
Obviously the squarefree part of the denominator of~$h$
divides~$xde$. Applying Hermite reduction and polynomial reduction to~$h$
then yields
\[
  h = \tilde g_1' + \frac{1}{d} R_1W + \frac{1}{x^\lambda e} Q_1V,
\]
where $R_1, Q_1\in K[x]^n$ with $\deg_x(R_1) < \deg_x(d)$ and $Q_1\in N_V$.
Repeating this discussion, we get the following lemma.
\begin{lemma}\label{LEM:idtf}
For any $i\in \bN$, the derivative $\partial_t^i\cdot f$ has an additive decomposition of the form
\[ \partial_t^i\cdot f = g_i' + \frac{1}{d} R_iW + \frac{1}{x^\lambda e} Q_iV,\]
where $g_i\in A$, $R_i, Q_i\in K[x]^n$ with $\deg_x(R_i) < \deg_x(d)$ and $Q_i\in N_V$.
\end{lemma}
As application of the above lemma, we can compute the minimal telescoper for $f$ by finding the first
linear dependence among the $(R_i, Q_i)$ over~$K$. We also obtain an upper bound for the order of telescopers.
\begin{corollary}
Every $f\in A$ has a telescoper of order at most $n\deg_x(d) + \dim_K(N_V)$.
\end{corollary}

\begin{example}
We compute a minimal telescoper for the function
\[
  F := \frac{1}{x^2} \log\Bigl(\frac{1}{x^2}-t^2\Bigr)\sqrt{\frac{1+tx}{1-tx}}.
\]
Note that for $t=1$ we
obtain the function from Example~\ref{ex:hr}. The operator~$L$ with $LF=0$
and the integral basis $\{\omega_1,\omega_2\}$ for
$A=\bC(t,x)[\partial_x]/\<L>$ are very similar to those in Example~\ref{ex:hr}.
As before, $F$ is represented by $f=1\in A$.
Also the computation of its Hermite reduction is analogous, yielding
\[
  h = \frac{-t^3x^2-t^2x+3t}{(t^2x^2-1)x} \omega_1 - \frac{t}{(t^2x^2-1)x} \omega_2
\]
as the Hermite remainder~$h$. The matrix~$M$ representing the differentiation
of the $\omega_i$ does not satisfy the degree condition of
Proposition~\ref{PROP:finite}; this fact is already visible
in~\eqref{eq:dmat}. Hence we perform a change of basis to
$\nu_1=x^{-1}\omega_1$, $\nu_2=x^{-2}\omega_2$, which is an integral basis at
infinity. We have $xeV'=BV$ for $V=(\nu_1,\nu_2)^T$, $e=x(t^2x^2-1)$, and
some matrix~$B\in K[x]^{2\times2}$ with $\deg_x(B)=3$; since
$\delta=\lambda+\deg_x(e)-1=3$ we are in Case~2 of
Proposition~\ref{PROP:finite}.  By investigating the eigenvalues of $\lc(B)$
we find that $\ell=1$. In order to determine a basis for~$N_V$ and to execute
the polynomial reduction conveniently, we consider the matrix whose rows are
constituted by $\phi_V(t)$ for all $t\in\cX_{0,\ell}$, written in the
basis~$\cX_{0,\ell+\delta}$. The echelon form of this $4\times10$ matrix is
\begin{equation}\label{eq:M_ell}
  \begin{pmatrix}
  1 & 0 & 0 & 0 & 0 & 2/t^3 & 0 & 4/t^4 & -4/t^4 & 0 \\
  0 & 0 & 1 & 1/t & 0 & 0 & 0 & -4/t^3 & 4/t^3 & 0 \\
  0 & 0 & 0 & 0 & 1 & 1/t & 0 & 0 & 0 & 0 \\
  0 & 0 & 0 & 0 & 0 & 0 & 1 & 0 & 0 & 0
  \end{pmatrix}.
\end{equation}
Writing the Hermite remainder~$h$ in terms of the $\nu_i$, we see that all
denominators divide the polynomial~$xe$. Thus we write
$h=\frac{1}{xe}(h_1\nu_1+h_2\nu_2)$ and in~\eqref{EQ:h} we get
$s_1=h_1=-t^3x^4-t^2x^3+3tx^2$, $s_2=h_2=-tx^3$, and $r_1=r_2=0$.
In the basis $\cX_{0,4}$ the vector $(s_1,s_2)$ reads
\[
  \bigl(-t^3, 0, -t^2, -t, 3 t, 0, 0, 0, 0, 0\bigr).
\]
The polynomial reduction now corresponds to reducing this vector with the rows
of~\eqref{eq:M_ell}, yielding the final result $[f]=-\frac{x^2}{xe}\nu_2$.  Next, we
consider the derivative $\partial_t\cdot f=tx^3\partial_x+2tx^2+x\in A$ whose
Hermite remainder is
\[
  \frac{1}{(tx+1)x} \omega_1 + \frac{1}{(t^2x^2-1)x} \omega_2.
\]
(Note that we could as well take $\partial_t\cdot [f]$ instead of $\partial_t\cdot f$, which
in general should result in a faster algorithm.)
After polynomial reduction we obtain
\[
  \frac{1}{xe} \biggl( -\frac{4}{t^2} \nu_1 + \frac{(tx+4)x}{t^2} \nu_2 \biggr).
\]
Since there is no linear dependence over $\bC(t)$ yet, we continue with
\[
  \partial_t^2\cdot f =
  \frac{(tx+1)x^3}{1-tx} \partial_x -
  \frac{x^2(2t^4x^4+5t^3x^3+2t^2x^2-5tx-3)}{(t^2x^2-1)^2}.
\]
Writing $\partial_t^2\cdot f$ in terms of the integral basis produces the denominator
$d=(tx-1)^3(tx+1)^2x$, which means that the Hermite reduction consists of three
reduction steps. As a Hermite remainder we obtain
\[
  \frac{t^2x^2+2tx-4}{(t^2x^2-1)tx} \omega_1 + \frac{2}{(t^2x^2-1)tx} \omega_2
\]
which by polynomial reduction is converted into
\[
  \frac{1}{xe} \Bigl( -\frac{4}{t^3} \nu_1 + \frac{2(tx+2)x}{t^3} \nu_2 \Bigr).
\]
Now we can find a linear dependence that gives rise to the telescoper
$t^2\partial_t^2-t\partial_t+1$, which is indeed the minimal one for this example.
\end{example}

\section*{Acknowledgements}

We would like to thank Ruyong Feng and Michael F.\ Singer for helpful discussions.

\bibliographystyle{elsarticle-harv}
\bibliography{Hermite}

\end{document}